\documentclass[10pt,times]{article}

\usepackage{fullpage}

\usepackage[cmex10]{amsmath}
\usepackage{amsfonts}
\usepackage{amssymb}
\usepackage{graphicx,color}
\usepackage{algorithm}
\usepackage{algorithmic}
\usepackage{latexsym}
\usepackage{subfigure}

\usepackage[noblocks]{authblk}

\newcommand{\eat}[1]{}

\newtheorem{definition}{Definition}[section]
\newtheorem{corollary}{Corollary}[section]
\newtheorem{theorem}{Theorem}[section]
\newtheorem{lemma}{Lemma}[section]

\newtheorem{proposition}{Proposition}[section]

\newtheorem{observation}{Observation}[section]

\newcommand{\qed}{\nopagebreak \hfill $\Box$}
\newenvironment{proof}{\par \noindent {\bf Proof}:}{\qed \par}

\newcommand{\beq}{\begin{equation}}
\newcommand{\eeq}{\end{equation}}
\newcommand{\baq}{\begin{eqnarray}}
\newcommand{\eaq}{\end{eqnarray}}
\newcommand{\baqm}{\begin{eqnarray*}}
\newcommand{\eaqm}{\end{eqnarray*}}
\newcommand{\barr}{\begin{array}}
\newcommand{\earr}{\end{array}}
\newcommand{\E}{\mathbb{E}}

\newif\ifabstract
\newif\iffull\fulltrue

\begin{document}

\title{Computing Traversal Times on Dynamic Markovian Paths}

\author{Philippe Nain\thanks{INRIA Sophia Antipolis, France} ~
Don Towsley\thanks{University of Massachusetts, Amherst, MA} ~ 
Matthew P. Johnson\thanks{University of California and Los Angeles, CA}\\
Prithwish Basu\thanks{Raytheon BBN Technologies, Cambridge, MA} ~
Amotz Bar-Noy\thanks{Brooklyn College, NY} $^{,\Vert}$ ~
Feng Yu\thanks{City University of New York Graduate Center, New York, NY}}

\makeatletter{\renewcommand*{\@makefnmark}{}

\date{}
\maketitle

%\thispagestyle{empty}

%%%%%%%%%%%%%%%%%%%%%%%
\begin{abstract}
In source routing, a complete path is chosen for a packet to travel from source to destination. While computing the time to traverse such a path may be straightforward in a fixed, static graph, doing so becomes much more challenging in dynamic graphs, in which the state of an edge in one time slot (i.e., its presence or absence) is random, and may depend on its state in the previous time step. The traversal time is due to both time spent waiting for edges to appear and time spent crossing them once they become available.

We compute the expected traversal time (ETT) for a dynamic path in a number of special cases of stochastic edge dynamics models, and for three edge failure models, culminating in a surprisingly challenging yet realistic setting in which the initial configuration of edge states for the entire path is known. We show that the ETT for this ``initial configuration" setting can be computed in quadratic time, by an algorithm based on probability generating functions. We also give several linear-time upper and lower bounds on the ETT. 
\end{abstract}

%\ifabstract

\vfill
\noindent \textbf{Acknowledgements.}
Research was sponsored by the Army Research Laboratory and was accomplished under Cooperative Agreement Number W911NF-09-2-0053. The views and conclusions contained in this document are those of the authors and should not be interpreted as representing the official policies, either expressed or implied, of the Army Research Laboratory or the U.S. Government. The U.S. Government is authorized to reproduce and distribute reprints for Government purposes notwithstanding any copyright notation here on.

\clearpage
%\fi
%%%%%%%%%%%%%%%%%%%%%%%
\section{Introduction}

In {source routing}, a complete path is chosen for a packet to travel, from source to destination, within a network~\cite{Argyraki04,Johnson01}. One potential advantage of source routing over dynamic routing is that once the routing path is chosen, no routing decisions need be made online. In the case of a fixed, static network graph, in which each edge is always available for use, the path's traversal time is simply the sum of the times to cross the constituent edges.
In a {\em dynamic graph} (modeling, for example, an ad hoc wireless network), edges may be intermittently unavailable. Specifically, the state of an edge in one timeslot (i.e., present or absent) may be random, as well as possibly dependent on its state in the previous timeslot.
The time spent traversing a route in such a dynamic graph includes both the time spent crossing edges and the time spent waiting for them to appear. Unlike in the static case, computing the expected value of this end-to-end traversal time within a dynamic graph is nontrivial. This is the problem we study in this paper.

% The problem we study in this paper is not routing (choosing the best path) but instead that of evaluating the cost of a specified path, in a dynamic graph. 
%In this paper we consider the problem of characterizing the end-to-end traversal time in a dynamic path graph whose connectivity is changing stochastically over time. 
We assume a discrete (slotted) model of time, over which edges appear and disappear; the state of an edge (on or off, or 1 or 0) in one timeslot depends on its state in the previous timeslot. In particular, the dynamics of each edge is governed by a Markov chain parameterized by $(p,q)$, the probabilities of an edge transitioning from off to on and on to off, respectively.
%
%Such dynamic path graphs have real-world applications in planning and supply-chain management, %(where supply links could obey stochastic dynamics), 
%communication along military convoys traveling through rugged/hostile terrain, and sensor network-based monitoring of linear civil structures such as bridges or trains.
%%, and disaster relief.
%Moreover, many routing schemes in wireless networks determine a path (e.g. a shortest path), and then remain on that path even though it may be intermittently connected due to its edges appearing and disappearing over time.

The expected traversal time (ETT) for a routing path with $n$ edges depends on $n, q, p$, the initial edge states, and the edge lengths. While the ETT can be straightforwardly estimated by simulation, that method suffers from high variance especially at low values of $q, p$; hence an exact characterization (algorithmic if not analytic) of the ETT is desirable. 
%In this paper, we characterize ETT for Markovian path graphs (with nonnegative integer edge lengths modeling forwarding delays) in a number of different stochastic settings which assume varying degrees of knowledge about the edge states, ranging from ``all edges are in steady state" to ``a given initial configuration of edge states".
We emphasize that 
%while the problem of dynamic routing collapses, on the path graph, to trivial, 
the problem of computing the ETT is surprisingly nontrivial. In fact, a highly restricted special case of the problem reduces to computing the highest order statistics of an IID sequence of geometric random variables, a problem with an analytic solution which itself was nontrivial to prove (originally by \cite{Szpankowski90}, later simplified by \cite{Eisenberg08}).
\eat{Indeed, the main topic of \cite{Eisenberg08} was computing $\E[\max_i\{X_i\}]$, where each $\{X_i\}$ is an IID sequence of geometric random variables with mean $1/p$, which is equivalent to a very special case of our problem:}
%\begin{corollary}[\cite{Eisenberg08,Szpankowski90}]
%Let $\hat n$ be the number of edges initially absent in the path graph, $q=0, p>0$, and all edge lengths 0. Then we have:
%\vspace{-1ex}
%\vskip -.05cm
%$$\text{ETT} = \sum_{i=1}^{\hat n-1} {\hat n-1 \choose i} (-1)^{i+1} / (1-(1-p)^i) = \Theta(\log \hat n)$$
%\end{corollary}
%%The general case of $(q> 0, p>0)$ is surprisingly more challenging.
%
%\begin{proof}
%In the stated special case, the traversal time is simply the time taken until all edges appear. Each edge's appearance time is an independent geometric random variable. It is known \cite{Eisenberg08} that the expectation of the max of $n$ such variables equals the stated value and order-of-magnitude.
%\end{proof}
\begin{corollary}[\cite{Eisenberg08,Szpankowski90}]
Let $\hat n$ be the number of edges initially absent in a path with n edges, $q=0$, $p>0$, and all edge lengths $0$. Then,
\vspace{-1ex}
\vskip -.05cm
$$\text{ETT} = \sum_{i=1}^{\hat n} {\hat n \choose i} \frac{(-1)^{i+1}}{1-(1-p)^i} = \Theta(\log \hat n)$$
\end{corollary}
%The general case of $(q> 0, p>0)$ is surprisingly more challenging.

\begin{proof}
In the stated special case, the traversal time is simply the time taken until all absent edges appear because present edges never disappear ($q=0$). Each edge's appearance time is an independent geometric random variable. It is known \cite{Eisenberg08} that the expectation of the max of $\hat n$ such variables equals the stated value and order-of-magnitude.
\end{proof}

%For the much more dififcult general Markov setting, we do not find a close-form expression for ETT but instead give a poly-time algorithm computing it.

%\vskip .1cm

%\vskip .1cm\noindent \textbf{Contributions.}
\subsection{Contributions}
Our main result is an exact $O(n^2)$-time algorithm for computing the ETT for the general $(q,p)$ model with edges starting in a specified initial configuration (see Corollary \ref{corr:expectation} of Sec. \ref{sec:initconds}), using an $O(n^2)$-time algorithm (see Theorem \ref{prop:Gnx}) to compute a family of {\em probability generating functions}. This algorithm applies to three edge failure models (see Sec. \ref{sec:models}).
We also compute the ETT for several special case settings (Sec. \ref{sec:simple}).
% and provide a number of easy-to-compute $O(n)$-time upper and lower bounds on the ETT (Sec. \ref{sec:bounds}).

\subsection{Challenges and techniques}
Designing a polynomial-time algorithm for the ETT in the general $(q,p)$ setting with known {\em initial} state requires overcoming several challenges. The ETT can be numerically approximated, assuming an algorithm to compute $\Pr(T=t)$, by $\sum_{t=0}^\tau t \Pr(T=t)$ for a large enough constant $\tau$. $\Pr(T=\tau)$ is nonzero for any arbitrarily large $\tau$, however, and we seek an exact solution.

A dynamic programming (DP) algorithm can be given to compute the ETT for a given initial state $s$ in terms of all possible next-timestep states (including $s$ itself), but there are exponentially many such states and hence subproblems to solve. 

Another natural strategy is to compute the expected time $ETT[i]$ to reach each node $i$ on the path: $ETT[i+1]$ depends not just on $ETT[i]$, but also on the probability that edge $(i,i+1)$ is present at the moment when node $j$ is reached (which {\em cannot} be assumed to equal $ETT[i]$). The state of $(i,i+1)$ at that point depends on its state at the previous point {\em and} on the state of $(i-1,i)$ at the previous point. Since these random states are not independent---there is eventually a large but polynomial number of subproblems---this leads to a complicated DP with running time $O(n^{11})$.
(Such an algorithm relies on transition probabilities of collections of edges changing from one state to another over the course of the random-duration process of waiting for the current missing edge to appear.)
%This probability also sums over an infinite number of possible outcomes (in the time dimension), but through algebraic manipulation, the ``infinite sum" portion can be made to converge to a finite quantity.

Instead, we apply probability and moment generating functions (see e.g. \cite{Grimmett} for an introduction) in order to obtain a much faster, quadratic-time algorithm.

\section{Preliminaries} \label{sec:models}

We begin with basic assumptions and concepts. Time is discrete, measured in time steps.
{\em Time t} refers to the beginning of the timeslot $t$ (numbered from 0).
Given is a path $n+1$ on nodes (0 to $n$) and $n$ edges.% ($G = K_n$ corresponds to the previous models.) 
%Natural interesting special cases for $G$ include path graphs, cycles, and trees. Here, we consider $G$ to be a path graph on $n+1$ nodes (denoted by $L_n$) in which the routing is from node 0 to node $n$.

%\begin{definition}
%{\em Dynamic $ER(p,G)$ graphs}: At each time $t$ there exists a graph $G_t$, which is the intersection of an $ER(n,p)$ random graph and $G$. An edge $e$ of $G$ is on (1) if it is present in $G_t$ and off (0) otherwise.
%\end{definition}

\begin{definition}
{\em Markovian $(q,p)$ paths}: At time 0, each edge is in some known state. The state of a given edge in subsequent timeslots is governed by a two-state Markov chain whose transition probabilities are given by $P(\text{off}\rightarrow \text{on})=p$, $P(\text{off}\rightarrow \text{off})=1-p$, $P(\text{on}\rightarrow \text{off})=q$, and $P(\text{on}\rightarrow \text{on})=1-q$.
\end{definition}

%\begin{observation}
%When $p=q$, this transition probability corresponds to stability level. With small $p$, we can expect few changes from $G_t$ to $G_{t+1}$. With large $p$, the graph will be less stable. An extreme case is the $(1,1)$ setting, the edge of $G$ alternates deterministically between on and off.
%\end{observation}

In the (1,1) setting, edge states alternate deterministically. In the $(1-p,p)$ setting, an edge's state is independent of its previous state.

\begin{definition}
Edge $e_i = (i-1,i)$ has length $d_i$, which may in general be a random variable (rv), with $D_i = \sum_{j=1}^i d_j$ and $D=D_n$. Edge lengths are all 0 in the {\em Cut-Through} ({\sf CuT}) model, all 1 in the {\em Store or Advance} ({\sf SoA}) model), and nonnegative integers in the {\em Distance} ({\sf Dist}) model.
\end{definition}

%\noindent \textbf{Transmission assumptions. }
Edge state indicates whether a packet can {\em begin} crossing the edge and, depending on the precise failure model, whether and when it will succeed.
If the packet is present at an edge's entry node {\em when the edge is on}, then the packet immediately traverses the edge; if it is currently off, the packet waits there until the edge appears. Edge transmission takes zero or more slots, depending on edge length. With length-zero edges, an unlimited number of contiguous on edges can be traversed instantly (modeling situations in which transmission times are negligible relative to time scales of disruption and repair \cite{Ram05}).
We consider three edge failure models.  In all cases, the packet requires some nonnegative number of time steps to cross a given link, if it is on. The models differ according to what occurs when the link fails prior to the completion of the transmission:

\begin{enumerate}
\item Transmission continues while the link is down; it simply cannot start unless the link is up. \label{mod:cantstart}
\item The remainder of the packet continues transmitting once the link returns.\label{mod:remainder}
\item The packet must be retransmitted on the link in its entirety.\label{mod:retransmit}
\end{enumerate}

Observe that the three models are equivalent in both {\sf CuT} and {\sf SoA}, but yield different behavior in {\sf Dist}. In model \ref{mod:remainder}, a transmission successfully completes once a total of $d_i$ {\em on} timeslots for edge $i$ occur; in model \ref{mod:retransmit}, a transmission completes once $d_i$ {\em on} timeslots for edge $i$ occur {\em in a row}.

%An edge $e_i$ turning off prevents packets from entering the edge but (for $d_i>1$) has no effect on packets currently crossing $e_i$.

%\vskip .1cm\noindent \textbf{Notation. }
A possible {\em state} or {\em configuration} of all edges is represented by a bitstring (e.g. $s$) of length $n$. $T$ is a random variable indicating time of arrival at node $n$ (e.g. $t$), with $\E[T] = ETT$. We sometimes write ETT$(s)$ to indicate ETT for initial state $s$, and $ETT[i]$ for the expected time of arrival at node $i$. We sometimes abbreviate Markov chain as MC.
%upper bound (UB), lower bound (LB).
%%%%%%%%%%%%%%%%%%%%%%%%%
\section{Three special settings} \label{sec:simple}

Before considering the problem in full generality, we consider three nice special settings.

%\noindent\textbf{The deterministic setting.}\label{sec:restrictedinitconds}
\subsection{The deterministic setting}\label{sec:restrictedinitconds}
Let the $(q,p)=(1,1)$, and assume each $d_i$ is a constant. Then the exact traversal time can easily be computed. First, consider model \ref{mod:cantstart}.
%Let the initial {\em configuration} of the path be represented by a binary string of length $n$. 
Let $b_i$ indicate the initial state of edge $(i-1,i)$.
%, and $d_i$ indicates the time to traverse edge $i$. If the $i$th bit is $1$ then the $i$th edge on the line exists otherwise it does not exist. For a given binary string $s$, let $k=k(s)$ be the number of changes from $0$ to $1$ or from $1$ to $0$. For example, $k(001110011001)=6$.
\begin{proposition}
Let $b_0=1$ and $d_0=0$. Let $\Delta_i = |b_{i+1} - b_{i}|$ indicate whether $b_{i+1} \ne b_i$, and let $k = \sum_{i=0}^{n-1} \Delta_i$. Then the traversal time in model \ref{mod:cantstart} is $D_n + \sum_{i=0}^{n-1} \left(d_i+\Delta_i \mod 2\right)$ in {\sf Dist}, $2n-k+1$ in {\sf SoA}, and $k$ in {\sf CuT}.
\end{proposition}
\begin{proof}
The traversal time is the sum of the total time spent crossing edges $D_n$ and the total time spent waiting. At each edge $i$, we wait one time step in two cases: $b_i$ differs from $b_{i-1}$ and $d_{i-1}$ was even, or $b_i = b_{i-1}$ and $d_{i-1}$ was odd.
In {\sf SoA}, this is $D_n = n$ plus $\sum_{i=0}^{n-1} \left(d_i+\Delta_i \mod 2\right)$, which is $\bar b_0$ plus the number of positions $i>0$ where $b_{i+1}=b_{i}$, or $n + n-k+1$.
In {\sf CuT}, we have  $D_n = 0$, and $\sum_{i=0}^{n-1} \left(d_i+\Delta_i \mod 2\right)$ is simply $k$.
\end{proof}

\begin{corollary}
The the traversal time {\sf Dist} in model \ref{mod:remainder} is $2D_n - k+1$.
\end{corollary}
\begin{proof}
Model \ref{mod:remainder} can be reduced to model \ref{mod:cantstart} by modeling each edge of size $d_i$ as a sequence of $d_i$ unit-size edges, all with the same initial state as $d_i$, for a total of $D_n$ such edges. Because of the chosen initial states, $k$ in the resulting instance will be the same as $k$ in the original.
\end{proof}

\vskip .15cm
With $d_i>1$ for some $i$, model \ref{mod:retransmit} does not apply to the deterministic setting because the packet will never succeed in crossing $d_i$.

%\vskip .25cm\noindent\textbf{The $(1-p,p)$ stochastic model.} 
\subsection{The $(1-p,p)$ stochastic model}
Computing the ETT here is straightforward.
\begin{proposition}\label{prop:SoACuT}
If the edge length for edge $(i-1,i)$ is given by a rv $d_i$, the expected routing time in model \ref{mod:cantstart} is $\sum_i E[d_i] + n(1-p)/p$ in {\sf Dist}, or $\frac{n}{p}$ and $\frac{n(1-p)}{p}$ in {\sf SoA} and {\sf CuT}, respectively.
\end{proposition}
\begin{proof}
The total expected transmission time is $\sum_i E[d_i]$. The expected wait for each edge to appear is $1/p$ if the edge is OFF (with probability $1-p$) and 0 otherwise.
\end{proof}

We note that when each $d_i$ is constant, $\Pr(T=t)$ for $t \ge D$ is the probability that $t-D$ times are spent waiting at nodes $0$ through $n-1$, analogous to throwing identical balls into bins:
$$\Pr(T=t) = \binom{t-D+n-1}{t-D} p^n(1-p)^{t-D}$$

%Let $N_t$ be a random variable denoting the node that the packet has reached at time $t$, and let $e_k$ be the $k$th edge. Then we have:
%\begin{eqnarray}
%\nonumber \Pr(N_t=k) & = & \Pr(N_{t-1}=k-1) \Pr(e_{k-1}) + \Pr(N_{t-1}=k) \Pr(\overline{e_k}) \\
%\label{eq:probdistER} & = & \Pr(N_{t-1}=k-1) p + \Pr(N_{t-1}=k) (1-p)
%\end{eqnarray}
%%
%Since depends on the number of ways of distributing $k$ pauses among time steps, we have:
%%
%$$ \Pr(N_{t-1})= \binom{t}{k}p^{k}(1-p)^{t-k}$$

%%%%%%%%%%%
%\noindent\textbf{The $(q,p)$ Markov model in steady state.}\label{sec:steady}
\subsection{The $(q,p)$ Markov model in steady state}\label{sec:steady}
%
%Now we study routing on dynamic line graphs $MC(q,p)$, where $p_0$ is the probability of an edge existing initially. To eliminate the effect of transients, in this section we assume that the Markov chain has converged (or {\em mixed}) before transmission begins. In this case we have:
%
%Now we consider $MC(q,p)$, assuming each Markov chain has converged (or {\em mixed}) before transmission begins. In this case we have:
If each  Markov chain has converged (or {\em mixed}) before transmission begins, we have:

\begin{proposition}
Let $p$ and $q$ both be nonzero. Then $MC(q,p)$ has a stationary distribution $\pi=(\pi_\text{on},\pi_\text{off})=(\frac{p}{p+q},\frac{q}{p+q})$.
\end{proposition}

As a corollary of Proposition \ref{prop:SoACuT} above, we then have:
\begin{corollary}\label{cor:CuTpq} \label{cor:SoApq}
If $d_i$ is again an rv, the expected routing time in model \ref{mod:cantstart} is $\sum_i E[d_i] + n(1-\pi_\text{on})/p$ in {\sf Dist}, or $n (1 + \frac{1-\pi_\text{on}}{p})$ and $\frac{n(1-\pi_\text{on})}{p}$ in {\sf SoA} and {\sf CuT}, respectively.
\end{corollary}
 
\eat{
Note that for the case $p=q=1$ we get here different results than those we had before. In {\sf SoA} we get $2n$ instead of $3n/2$ and in {\sf CuT} we get $n$ instead of $n/2+1$. The reason is that there we assumed a uniform distribution over all configurations and here we assumed a steady state configuration.
}

Fig. \ref{fig:segments} illustrates paths through space-time corresponding to the progress of the packet traveling from node $0$ to $n$, in {\sf CuT} and {\sf SoA}. 
Each such path is composed of segments of moving and waiting. Let $m$ be the number of moving segments, each preceded by a wait segment (possibly empty in the first case). Let the path be specified by a sequence (see Fig. \ref{fig:segments}) $\{0,0, k_1, t_1, k_2, t_2, \ldots, k_m, t_m,n,t\}$, and assume each $d_i$ is constant.
%Here all segments except $[t_0,t_1]$ are required to be of strictly positive length. 
Clearly $1\leq m \leq \min(n,t+1-D)$. Since edge transmission takes time $d_i$ the total latency $t$ obeys $t \geq D = \sum_e d_e$.
%
%Let $X_\tau^e$ be a binary random variable denoting the state of edge $e$ at time step $\tau$. 
Since the state of an edge $i$ over time is governed by a Markov chain, the probability of a {\em waiting} segment of length $\ell$ is $\pi_\text{off} \: (1-p)^{\ell-1} \: p$.
The probability of some path $P = \{(0,0), (k_1, t_1), (k_2, t_2), \ldots, (k_m, t_m),(n,t)\}$ conditioned on $k_1=0$ is given by the following:
\begin{eqnarray*}
\label{eq:mc3} \lefteqn{(\pi_\text{off} (1-p)^{t_1-1} p ) \pi_\text{on}^{k_1-1} ~ (\pi_\text{off} (1-p)^{t_2-t_1-1} p ) \pi_\text{on}^{k_2-k_1-1} \cdots (\pi_\text{off} (1-p)^{t-t_m-1} p) \pi_\text{on}^{n-k_m-1}}\\
\label{eq:mc4} & = & \pi_\text{on}^{n-m} ~ \pi_\text{off}^m ~ (1-p)^{t-D-m} ~ p^m\\ 
&=& \Big(\frac{p}{p+q}\Big)^{n-m} \Big(\frac{q}{p+q}\Big)^m (1-p)^{t-D-m} p^m~~~~~~~~~~~~~~~~~~~~~~~~~~~~~~~~~~~~~~~~~~~~~~\\
 &=& \frac{p^{n} q^m (1-p)^{t-D-m}}{(p+q)^{n}}
\end{eqnarray*}

\begin{figure*}[ttt!]
\centering
    \subfigure[{\sf CuT} model.]{\label{fig:timespace}\includegraphics[trim = 4cm 4cm 4cm 4cm, width=.48 \textwidth]{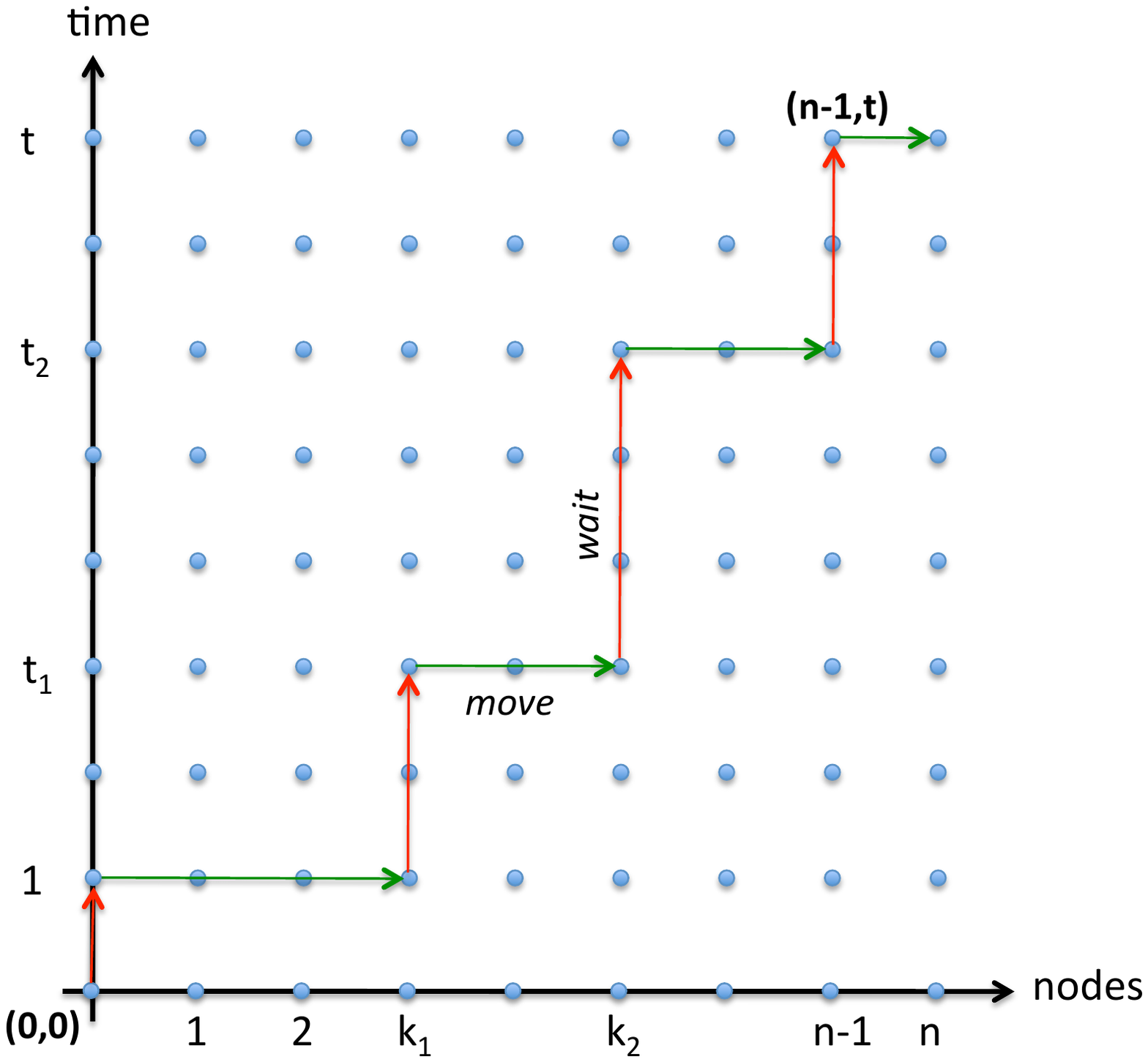}}
    \subfigure[{\sf SoA} model.]{\label{fig:timespace-soa}\includegraphics[trim = 4cm 4cm 4cm 4cm, width= .48\textwidth]{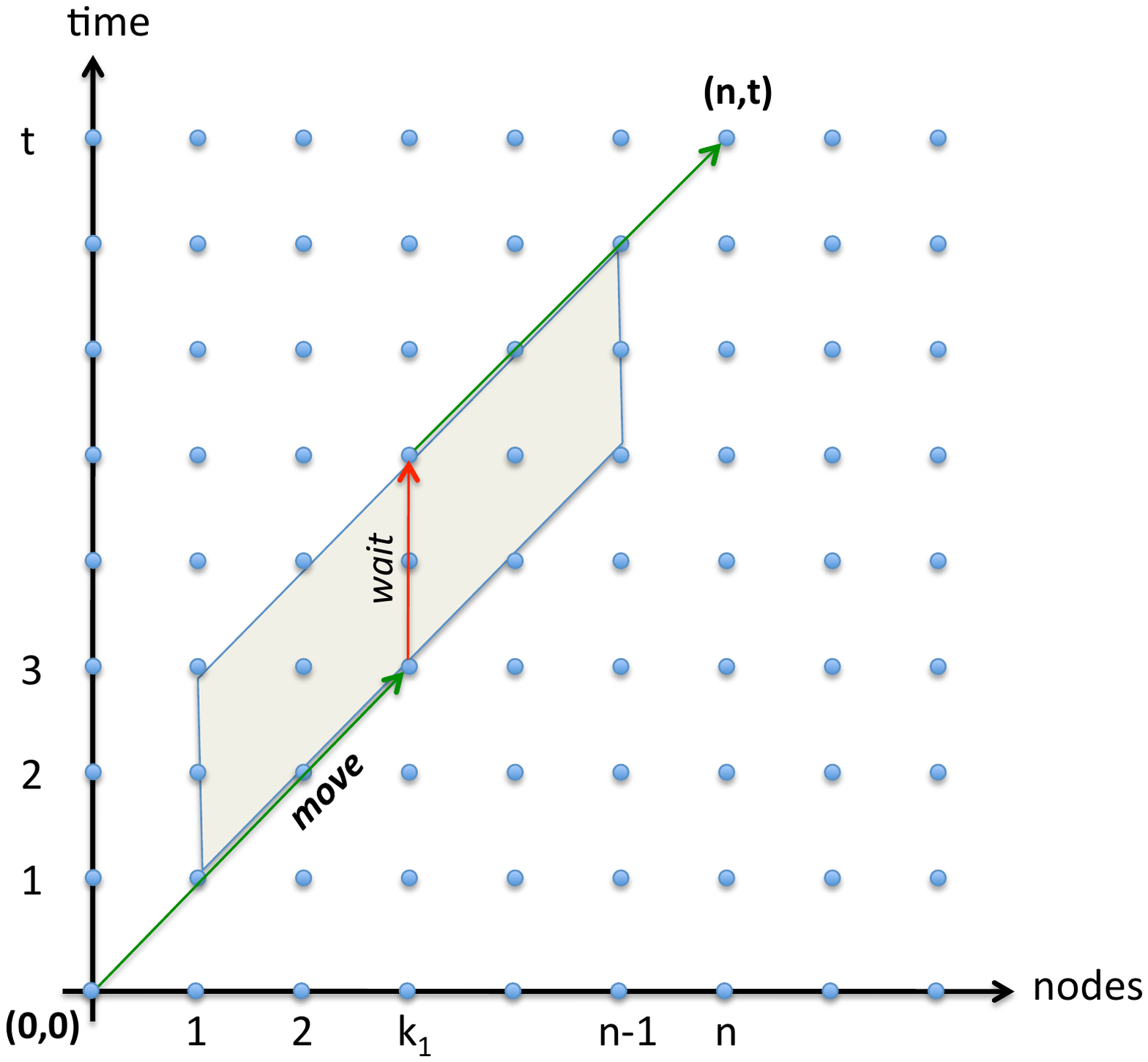}}

\caption{Analyzing the latency in Markov path graphs.}\label{fig:segments}
\end{figure*}
%Eq. \ref{eq:mc2} follows 
%from Eq. \ref{eq:mc1} by using 
%from the fact that edge states are independent of one another. 

The probability of path $P$ conditioned on $k_1>0$ is the same.
A path with exactly $m$ right-bends can be generated by independently choosing $m-1$ right-bending points on the space axis (the first is at $k=0$) and $m$ points on the time axes, ${n-1\choose m-1}$ and ${t\choose m}$ ways respectively, for a total of ${n-1\choose m-1}{t\choose m}$ such paths. Then the latency probability distribution $\Pr(T=t)$ for $p>0,q>0$ is given by:

\begin{equation}\label{eq:mcprob}
%\lefteqn{\Pr(T=t-1)} \nonumber\\
\Pr(T=t) = \sum_{m=1}^{\min\{n,t+1\}} {n-1\choose m-1}{t-D\choose m} \frac{p^{n} q^m (1-p)^{t-D-m}}{(p+q)^{n}}
\end{equation}
\eat{
\\ & = & (n-1) \frac{p^{n-1} q (1-p)^{t-2}}{(p+q)^{n-1}} {}_2F_1 (n-2,t-2;2;\frac{q}{1-p})\nonumber
\noindent where ${}_2F_1 (a,b;c;z)$ is the hypergeometric function:
\[ {}_2F_1 (a,b;c;z) = 1 + \frac{ab}{1! c}z + \frac{a(a+1)b(b+1)}{2! c(c+1)}z^2 + \cdots \]
}

%%%%%%%%%%%%%%%%%%%%%%%%%
\section{The Markov model with initial configuration}\label{sec:initconds}

%
%Consider ${n}$ mutually independent Markov on/off links in series.
Let $X_i(t)\in \{0,1\}$ be the state of link $i$ at time $t$,
where $X_{i}(t)=1$ (resp. $X_{i}(t)=0$) if link ${i}$ is on (resp. off)
at time $t$, and let $X(0)=x=(x_1,\ldots,x_{n})\in \{0,1\}^{n}$ be the initial link state, at time 0. The probability transition matrix $P_i$ for link ${i}$ is given by 
\[ P = \left( \begin{array}{ll} 1-p & p \\ q & 1-q \end{array} \right) \]
with $p$ (resp. $q$) the transition probability that link ${i}$ jumps from state $0$ (resp. state $1$) into state $1$
(resp. state $0$) in one timestep. 
Let $P_{a,b}({i},t)=\Pr(X_{i}(t)=b\,|\,X_{i}(0)=a)$ be the probability that link ${i}$ is in state $b\in \{0,1\}$ at time $t$ given that
it was in state $a\in \{0,1\}$ at time $t=0$.  Note that $P_{a,b}(t)$ does not depend on the link's identity since all links are assumed to have the same parameters $p$ and $q$.
Let $\beta =1 - p - q \in [-1,1]$. It is known \cite{Levin08} that
\begin{eqnarray}
P_{1,0}(t) & = & \pi_0 (1 - \beta^t),\quad P_{1,1}(t)=\pi_1 +\pi_0 \beta^t \label{transient1}\\
P_{0,1}(t) & = & \pi_1 (1 - \beta^t),\quad  P_{0,0}(t)=\pi_0 +\pi_1 \beta^t  \label{transient2}
\end{eqnarray}
for any $t\geq 0$, where $\pi_{0}=q/(p + q)$ and $\pi_{1}=p/(p + q)$ are the stationary probabilities that link $i$  is in states $0$ and in $1$, respectively.

Let $T_{i}$ be the time at which the packet reaches link ${i}$, i.e. the time at which it finishes crossing link ${i}-1$. Let $D_i = T_{{i}+1}-T_{i}$ be the time spent waiting for link ${i}$ to appear plus the time taken to cross it.
To compute $\E[T_{n}]$, we begin by computing the probability generating function (PGF) of $T_i$, namely, $G_{{i},x}(z)= \E[z^{T_{i}} | X(0)=x]$ and for initial link state $x=(x_1,\ldots,x_{n})$ and ${i}={n}$. Let
\[
F_{i,1}(z) = \E\left[z^{D_{i}}\,|\,X_{i}(T_{{i}-1})=1\right], \quad F_{i,0}(z) = \E\left[z^{D_{i}}\,|\,X_{i}(T_{{i}-1})=0\right]
\]
be the probability generating function (PGF)  of $D_{i}$ given that link ${i}$ was in state $1$ (resp. state $0$) at time $T_{{i}-1}$, respectively, and define for all links $i\in\{1,2,\ldots,n\}$,
\[
\gamma_{i,1}=\E[D_{i}\,|\,X_{i}(T_{{i}-1})=1], \quad \gamma_{i,0}=\E[D_{i}\,|\,X_{i}(T_{{i}-1})=0]
\]

Note that $\gamma_{i,1} = dF_{i,1}(z)/dz|_{z=1}$ and  $\gamma_{i,0} = dF_{i,0}(z)/dz|_{z=1}$.

The computation for the three edge failure models given in Sec. \ref{sec:models} will differ only according to the expressions for $F_{i,0}(z)$ and $F_{i,1}(z)$, which we derive expressions for below.
For any number $a\in [0,1]$ let $\bar a =1-a$. Note that expressions of the form $a\alpha + \bar a \beta$ are equivalent to $\alpha \textbf{~if~} a=1 \textbf{~else~} \beta$.

\subsection{Computing the PGFs and the ETT}

\begin{theorem}[Computing $G_{{i},x}(z)$]\label{prop:Gnx}
For initial state $X(0)=x=(x_1,\ldots,x_{n})$, we can, given the values $F_{i,0}(\beta^{i}),F_{i,1}(\beta^{i})$ for ${i}$ from 0 to ${n}-1$ (each of which can be found in constant time by the computations in the following section), compute the following in $O({n}^2)$ time:\begin{eqnarray*}
G_{1,x}(z)&=&x_1 F_{1,1}(z)+\bar x_1 F_{1,0}(z) \label{MGF-Gn0} \\
G_{{i},x}(z)&=&\phi_i(z) \cdot G_{{i}-1,x}(z) + \chi({i}) \cdot \psi_i(z) \cdot G_{{i}-1,x}(\beta z),\quad {i}=2,\ldots,{n}  \label{MGF-Gn}
\end{eqnarray*}
where
\begin{equation}
\label{def-phi-psi}
\phi_i(z)=\pi_0 F_{i,0}(z)+\pi_1 F_{i,1}(z),\quad \psi_i(z)=F_{i,0}(z)-F_{i,1}(z), \quad \chi({i}) = \bar x_{i}\pi_1 - x_{i}\pi_0
\end{equation}
\end{theorem}

\iffull
\begin{proof}
%Let $x=(x_1,\ldots,x_N)$ be fixed.
We prove correctness recursively. The base case ($i=1$) follows immediately from the fact that $T_1=D_1$.  We now show, for expository reasons, that (\ref{MGF-Gn0}) holds for ${i}=2$. We have:
\begin{eqnarray}
G_{2,x}(z)&=&\sum_{t=0}^\infty z^t  \E\bigl[z^{D_2}|X(0)=x, D_1=t\bigr] \Pr(D_1 = t | X(0)=x)\nonumber\\
&=&
\sum_{t=0}^\infty z^t  \sum_{j\in \{0,1\}}\E\left[z^{D_2}|X(0)=x, D_1=t, X_2(t)=j\right]\nonumber \times \Pr(X_2(t)=j|X(0)=x, D_1=t)\\
&& \mbox{} \times \Pr(D_1 = t  |X(0)=x)\nonumber\\
&=&
\sum_{t=0}^\infty z^t  (F_{2,0}(z) \Pr(X_2(t)=0|X(0)=x, D_1=t) + F_{2,1}(z) \Pr(X_2(t)=1|X(0)=x, D_1=t) )\nonumber \\
& & \mbox{} \times \Pr(D_1 = t  |X(0)=x)  \nonumber\\
&=&
\sum_{t=0}^\infty z^t  \phi_2(z)  \Pr(D_1 = t  |X(0)=x)  \nonumber + \sum_{t=0}^\infty z^t \beta^t \cdot \chi(2) \cdot \psi_2(z)   \Pr(D_1 = t |X(0)=x) \nonumber \\
 & = & \phi_2(z) \cdot G_{1,x}(z) + \chi(2) \cdot \psi_2(z) \cdot G_{1,x}(\beta z) \label{int0}
\end{eqnarray}
where we have used (\ref{transient1}) and (\ref{transient2}) to derive (\ref{int0}).

Assume for induction that (\ref{MGF-Gn0}) holds for ${i}=2,\ldots,m$. We now show that it holds for ${i}=m+1$:
\begin{eqnarray}
\lefteqn{G_{m+1,x}(z)} \nonumber\\
&=&\sum_{t_1=0}^\infty \cdots \sum_{t_m=0}^\infty
z^{\sum_{i=1}^m t_i} \E\left[z^{D_m}|X(0)=x, D_1=t_1, \ldots,D_m=t_m\right]\nonumber\times  \Pr(D_1=t_1,\ldots, D_m =t_m  |X(0)=x)\nonumber\\
&=&
\sum_{t_1=0}^\infty \cdots \sum_{t_m=0}^\infty
z^{\sum_{i=1}^m  t_i} \sum_{j\in \{0,1\}}
\E\left[z^{D_m}|X(0)=x, D_1=t_1, \ldots,D_m=t_m, X_{m+1}\Big(\sum_{k=1}^m t_k\Big) = j\right]
\nonumber\\
&& \times  \Pr(X_{m+1}\Big(\sum_{k=1}^m t_k\Big)=j| X_{m+1}(0)=x_{m+1}) \cdot \Pr(D_1=t_1,\ldots, D_m=t_m |X(0)=x)\nonumber\\
&=&
\sum_{t_1=0}^\infty \cdots \sum_{t_m=0}^\infty z^{\sum_{i=1}^m  t_i}
  (\phi_{m+1}(z) + \beta^{\sum_{i=1}^m  t_i} \cdot \chi(m+1) \cdot \psi_{m+1}(z))\nonumber\\
&& \times \Pr(D_1=t_1,\ldots, D_m=t_m |X(0)=x) \label{12again} \\
&=&
 \phi_{m+1}(z) \cdot G_{m,x}(z) + \chi(m+1) \cdot \psi_{m+1}(z) \cdot G_{m,x}(\beta z) \label{gm1}
\end{eqnarray}
where (\ref{12again}) follows from (\ref{transient1}) and (\ref{transient2}) and (\ref{gm1}) follows from the definition of PGF.

For complexity, observe that we must compute $G_{{i},x}(\beta^k)$ for ${i}$ from 1 to ${n}$ and $k$ from 0 to ${n}-1$, and that each such computation is done in constant time.
\end{proof}
\fi

\begin{corollary}[Computing ETT]
\label{corr:expectation}
For initial state $x=(x_1,\ldots,x_{n})$, we can (again given the $O(1)$-time computable values $F_{i,0}(\beta^{j}),F_{i,1}(\beta^{j})$ for ${i}$ from 1 to ${n}$ and $j$ from 0 to $n-1$) compute ETT in $O({n}^2)$ time:
\begin{equation}
\label{prop:TN}
\overline T_{n}(x)= \sum_{{i}=1}^{{n}} \pi_0 \gamma_{i,0}+ \pi_1\gamma_{i,1}+(\gamma_{i,0}-\gamma_{i,1}) \cdot \chi(i) \cdot G_{i-1,x}(\beta)
\eat{{n}(\pi_0 \gamma_0+ \pi_1\gamma_1) +(\gamma_0-\gamma_1)\sum_{{i}=1}^{{n}} \chi({i}) \cdot G_{{i}-1,x}(\beta)}
\end{equation}
\end{corollary}

\iffull
\begin{proof}
Computing $G_{{i},x}(z)$ for all ${i}$ is done by Theorem \ref{prop:Gnx} in $O({n}^2)$ total time. The only additional computation time here is $O({n})$.
From this, any statistic of $T_{i}$ can be obtained.  Consider:
\begin{equation}
\label{def:expectation}
\overline T_{i}(x)=\E[T_{i} | X(0)=x]
\end{equation}
the expected transmission time on links $1,\ldots,{i}$ given that the system is in state $x$ at
time $t=0$.

Observe that $\overline T_{i}(x)=G^\prime_{{i},x}(1)$. (By convention, $G_{0,x}(z) = 1$.) We then obtain the recursion (using the fact that $\phi_i(1)=\pi_0+\pi_1=1$, $\psi_i(1)=F_{i,0}(1)-F_{i,1}(1)=0$, and $G_{i-1}(1)=1$):
\begin{eqnarray*}
\overline T_1(x)&=&x_1 \gamma_{1,1} + \bar x_1 \gamma_{1,0}\\
\overline T_{i}(x)&=&\overline T_{{i}-1}(x)+ \pi_0 \gamma_{i,0}+ \pi_1\gamma_{i,1}
+(\gamma_{i,0}-\gamma_{i,1}) \cdot \chi(i) \cdot G_{i-1,x}(\beta),\quad i=2,\ldots,{n}
\end{eqnarray*}
which in closed form is (\ref{prop:TN}).
\end{proof}
\fi

\begin{observation}[Computing the probability distribution]
Since the entire probability distribution of $T_n(x)$ is captured by the PGF, $G_{n,x}(z)$, it is possible to retrieve the probability distribution of $T_n(x)$ by taking repeated derivatives as follows~\cite{Grimmett}:
\[ \Pr(T_n(x)=k) = \frac{G_{n,x}^{(k)}(0)}{k!} \]
While computing these may not be feasible in closed form in general, it is easy to perform this computation numerically.
\end{observation}

\subsection{Computing $F_{i,1}(z)$ and $F_{i,0}(z)$}

Now that the PGF of $T_i$ has been computed, in order to solve various link failure models, we need to compute the PGFs of $D_i$ as well, i.e. $F_{i,1}(z)$ and $F_{i,0}(z)$.
Let $S$ be the random variable (rv) denoting the time needed to traverse a link and $Y$ be the rv denoting the amount of time spent waiting, upon arrival there, for the link to turn on.
First observe that for all links $i\in\{1,2,\ldots,n\}$: 
\begin{equation}
\label{F1F0}
F_{i,0}(z)= G_Y(z) F_{i,1}(z)
\end{equation}
since $D_{i}|_{X_{i}(T_{{i}})=0} =_d Y+ D_{i}|_{X_{i}(T_{{i}})=1}$, ($=_d$ indicates equality in distribution)
where $Y$ is a geometrically distributed rv with parameter $p$, independent of $ D_{i}|_{X_{i}(T_{{i}})=1}$. (A sum of rv's yields a product of PGFs.) In our model, since $Y$ corresponds to the duration of an off period of a link, $\Pr(Y=k)=(1-p)^{k-1}p$ which yields the following PGF:
\begin{equation}
\label{GYz}
G_Y(z) = 
        \frac{pz}{1-(1-p) z}, \quad 0\le |z|\le 1
\end{equation}

%We now prove a lemma that will be useful in Sec. \ref{sec:bounds} below.
\begin{lemma}\label{lem:gamma}
For the case of geometrically distributed link off periods (Y), $\gamma_0=\gamma_1 +\frac{1}{p}$
\end{lemma}

\iffull
\begin{proof}
Recall that $\gamma_{i,1} = dF_{i,1}(z)/dz|_{z=1}$ and  $\gamma_{i,0} = dF_{i,0}(z)/dz|_{z=1}$ and $ F_{i,0}(z)= G_Y(z) F_{i,1}(z)$.
It is easy to see that $G_Y(1)=1$ and $d G_Y(z)/dz |_{z=1} =\frac{1}{p}$, whereas $F_{i,1}(1) = \E\left[1\,|\,X_{i}(T_{{i}})=1\right]=1$. In this case, we have:
\begin{eqnarray*}
\gamma_{i,0}&=&d F_{i,0}(z)/dz |_{z=1}\\
&=&d G_Y(z)/dz|_{z=1} * F_{i,1}(1) + G_Y(1)*d F_{i,1}(z)/dz |_{z=1}\\
&=& 1/p + \gamma_{i,1}
\end{eqnarray*}
\end{proof}
\fi

\subsubsection{Failure model \ref{mod:retransmit}: complete retransmission after link failure}

%In this scenario the file has to be entirely retransmitted if the link moves to the off state
%during a transmission.  

We consider two subcases:
\begin{itemize}
\item[(a)] Successive retransmission times on a link are all {\em identical}, denoted by the rv $S$.
\item[(b)] Successive retransmission times on a link are iid rvs $\{S,S_1,S_2\ldots,\}$.
\end{itemize}

Let $G_S(z)$ be the PGF of $S$. \iffull Throughout $\{W, W_1,W_2,\ldots\}$ and $\{Y, Y_1,Y_2,\ldots\}$ are assumed to be
independent sequences of iid rvs, 
furthermore independent of  $S$ in case (a) and of  $\{S,S_1,S_2\ldots,\}$ in case (b),
that are geometrically distributed with parameters $q$ 
and $p$, respectively. The rv $W_i$ (resp. $Y_i$) corresponds to the $i$th on period (resp. off period) of a link
since the first attempt to transmit the packet. 
Let ${\bf 1}_A$ be the indicator function of the event $A$. \fi Also, in the subsequent text, the subscript $i$ is implicit in the expressions involving $F_0(z), F_1(z)$ for notational convenience.

\paragraph{Case (a):}
\iffull
Conditioned on $X(T_{{i}})=1$, we have
\[
D_{i}=_d \sum_{i\geq 0} {\bf 1}_{A_i(S)}\left(S+\sum_{l=1}^i (W_l + Y_l)\right)
\]
with $A_i(u)=\{W_1<u,\ldots, W_i<u,W_{i+1}\geq u\}$. Note that $A_i(u)\cap A_j(u) =\emptyset$ for $i\not=j$ 
and $\sum_{i\geq 0} \Pr(A_i(u))=1$, for any $u$. We have
\baqm
F_1(z) & = & \sum_{u=0}^\infty \E\left[z^{\bigl(  \sum_{i\geq 0} {\bf 1}_{A_i(u)}\bigl(u+\sum_{l=1}^i(W_l+Y_l)\bigr)\bigr)}\,\Bigl|\,S=u \right] \Pr(S = u) \\
&=& \sum_{u=0}^\infty \Pr(S=u) z^{ u} \sum_{i= 0}^\infty \E\left[z^{\bigl(\sum_{l=1}^i(W_l+Y_l)\bigr)}\,|\,S=u, A_i(u)=1\right] \Pr(A_i(u)=1|S=u)  
\\
&=&
\Pr(S=0)+\sum_{u=1}^\infty \Pr(S=u)z^u \Pr(W\ge u)  \sum_{i=0}^\infty   G_Y(z)^i \E[z^W|W<u]^i \Pr(W<u)^i\\
&=&
\Pr(S=0)+\frac{1}{1-q} \sum_{u=1}^\infty  \frac{ \Pr(S=u)((1-q)z)^{u}}{1 - G_Y(z)\E[z^W|W<u]\Pr(W<u)}.
\eaqm
The latter identity follows from $\Pr(W=i)=(1-q)^{i-1} q$, $i\geq 1$, which yields
\[
\Pr(W\geq u)=q\sum_{i=u-1}^{\infty} (1-q)^{i}=
q\left(\sum_{i=0}^\infty (1-q)^i -\sum_{i=0}^{u-2}(1-q)^i\right)=(1-q)^{u-1},\quad u\geq 1
\]
On the other hand, for $u \geq 1$,
\begin{eqnarray*}
\lefteqn{\E[z^W|W<u]\Pr(W<u)= \sum_{i=1}^{u-1} E[z^W|W<u,W=i]\Pr(W=i|W<u)\Pr(W<u)}\\
&=&\sum_{i=1}^{u-1} z^i \Pr(W=i) =qz  \sum_{i=0}^{u-2} ((1-q)z)^{i} = qz  \frac{1- ((1-q)z)^{u-1}}{1-(1-q)z},
\end{eqnarray*}
so that
\fi
\begin{eqnarray*}
F_1(z)&=& \Pr(S=0)+\frac{(1-(1-p)z)(1-(1-q)z)}{1-q}\\
&&\times  \sum_{u=1}^\infty  \frac{ \Pr(S=u)((1-q)z)^{u}}{(1-(1-p)z)(1-(1-q)z) -pq z^2(1-((1-q)z)^{u-1})}
\end{eqnarray*}
by using (\ref{GYz}). From (\ref{F1F0}) we deduce
\begin{eqnarray*}
F_0(z) &=&  \Pr(S=0)+pz\, \frac{1-(1-q)z}{1-q}\times  \sum_{u=1}^\infty  \frac{ \Pr(S=u)((1-q)z)^{u}}{(1-(1-p)z)(1-(1-q)z) -pq z^2(1-((1-q)z)^{u-1})}
\end{eqnarray*}

If $S=1$ (i.e., {\sf SoA}), we have
\baqm
F_0(z) =  \frac{pz^2}{1-(1-p)z}, \quad
F_1(z) = z
\eaqm

And if $S = 0$ (i.e., {\sf CuT}), we have
\baqm
F_0(z) =  \frac{pz}{1-(1-p)z}, \quad
F_1(z) = 1
\eaqm

\paragraph{Case (b):}
\iffull
Conditioned on $X(T_{{i}})=1$, we have
\[
D_{i}=_d \sum_{i\geq 0} {\bf 1}_{B_i(S_1,\ldots,S_{i+1})}\left(S_{i+1}+\sum_{l=1}^i(W_l+Y_l)\right)
\]
with $B_i(u_1,\ldots,u_{i+1})=\{W_1<u_1, \ldots, W_i<u_i,W_{i+1}\geq u_{i+1}\}$. Hence,
\begin{eqnarray}
F_1(z) & = & 
  \sum_{i= 0}^\infty \Pr(B_i(S_1,\ldots,S_i)=1) \E\left[z^{(S_{i+1}+\sum_{l=1}^i(W_l+Y_l))}\,|\,B_i(S_1,\ldots,S_i)=1\right]
 \nonumber \\
&=&
\sum_{i= 0}^\infty \Pr(W<R)^i \Pr(W \ge S) Y(z)^i 
\E\left[z^{(S_{i+1}+\sum_{l=1}^i W_l)}\,|\,S_1>W_1,\ldots,S_i>W_i, S_{i+1}\le W_{i+1} \right]  \nonumber\\
& = & 
\sum_{i= 0}^\infty G_Y(z)^i (\Pr(W<S) \E[z^W|W<S])^i \Pr(W \ge S)\E[z^S|W\ge S] \nonumber \\
& = &  \frac{ \Pr(W \ge S)\E[z^S|W\ge S]}{1-G_Y(z) \Pr(W<S) \E[z^W|W<S]} \label{eq:F1caseb}
\end{eqnarray}

Let us concentrate on $\E[z^S|W\ge S]\Pr(W\ge S)$ and $\E[z^W|W<S]\Pr(W<S)$:
\baqm
\E[z^S|W\ge S]
& = & \sum_{u=0}^\infty z^u \Pr(S=u|W\ge S) = \frac{\Pr(S=0)+ \sum_{u=1}^\infty z^u \Pr(S=u)\Pr( W\ge u)}{\Pr(W\ge S)} \\
& = & \frac{\Pr(S=0) + \sum_{u=1}^\infty z^u \Pr(S=u)(1-q)^{u-1}}{\Pr(W\ge S)} = \frac{(G_S((1-q)z) - q\Pr(S=0))}{(1-q)\Pr(W\ge S)}
\end{eqnarray*}
or 
\[ \E[z^S|W\ge S]\Pr(W\ge S) =  \frac{(G_S((1-q)z) - p\Pr(S=0))}{(1-q)} \]
and
\baqm
\E[z^W|W<S]\Pr(W<S) &=& \Pr(W<S) \sum_{w=1}^\infty z^w \Pr(W=w|W<S) \\
& = & \sum_{w=1}^\infty z^w \Pr(W=w)\Pr(S>w) = qz\sum_{w=1}^\infty ((1-q)z)^{w-1} \sum_{u=w+1}^\infty \Pr(S=u) \\
& = & qz\sum_{u=2}^\infty \Pr(S=u) \sum_{w=1}^{u-1} ((1-q)z)^{w-1} =  qz\sum_{u=2}^\infty \Pr(S=u) \frac{1-((1-q)z)^{u-1}}{1-(1-q)z} \\
& = & \frac{q(G_S((1-q)z)-\Pr(S=0)-\Pr(S=1)(1-q)z))}
           {(1-q)(1-(1-q)z)}
\eaqm
Substituting into (\ref{eq:F1caseb}) yields:
\fi
\[
F_1(z)  = 
 \frac{(G_S((1-q)z) - q\Pr(S=0))(1-(1-q)z)(1-(1-p)z)}
    {(1-(1-p)z)(1-(1-q)z)(\Pr(S=0) + \Pr(S = 1)(1-q)z - G_S((1-q)z))}
\]

\subsubsection{Failure model \ref{mod:remainder}: transmission is resumed after link failure}
\iffull Let us compute $F_1(z)$. Conditioned on $X(T_{{i}})=1$, we have $D_{i} = S +  \sum_{l=1}^V Y_l$,
where $V$ is a binomial rv with parameter $q$ and population $S-1$. Note that given a population $S= u+1$,  $G_V(z) = 1-q(1-z))^u$. If a sum of rv's $Z=X_1+X_2+\cdots+X_N$ where $X_i$'s are iid rv's and $N$ is also a rv, then $G_Z(z)=G_N(G_X(z))$~\cite{Grimmett}. Hence we can write: \fi
\baqm
F_1(z) & = & 
  \Pr(S=0) + \sum_{u=1}^\infty \Pr(S=u) z^u [1-q(1-G_Y(z))]^{u-1} \\
& = & \frac{G_S(1-q(1-G_Y(z))) - \Pr(S=0)q(1-Y(z))}{1-q(1-G_Y(z))}\\
F_0(z) &=& G_Y(z) \frac{G_S(1-q(1-G_Y(z))) - \Pr(S=0)q(1-G_Y(z))}{1-q(1-G_Y(z))}
\eaqm

\subsubsection{Failure model \ref{mod:cantstart}: already started transmission is unaffected by link failure}
This is by far the simplest scenario:
\baqm
F_1(z) = G_S(z), \quad
F_0(z) = G_Y(z) G_S(z)
\eaqm

%%%%%%%%%%%%%%%%%%%%%%%%%
%\input{qp-approx}
%%%%%%%%%%%%%%%%%%%%%%%%%
%%%%%%%%%%%%
\section{Discussion and future work}
\label{sec:discuss}

In this paper we exactly computed the expected time to traverse a dynamic path with edge states governed by Markov chains. Natural interesting generalizations include edge failures that are not independent; for example, adjacent pairs of edge failures would correspond to a node failure, and  probability $q_i,p_i$ that vary by link. Our algorithms maintain the same complexity as long as $p_i+q_i$ is constant, but in the full generalization  become exponential-time. These techniques have applications in modeling communication along military convoys traveling through rugged terrain, and sensor network-based monitoring of linear civil structures such as bridges or trains.

\ifabstract
\clearpage
\fi

\vspace{-2ex}
\bibliographystyle{abbrv}
\bibliography{temporalgraph,shortpath,bib}

\end{document}